\documentclass[conference, 10pt]{IEEEtran}
\IEEEoverridecommandlockouts
\usepackage{caption}
\usepackage{subcaption}
\usepackage{cite}
\usepackage{amsmath,amssymb,amsfonts}
\usepackage{graphicx}
\usepackage{textcomp}
\usepackage{xcolor}
\def\BibTeX{{\rm B\kern-.05em{\sc i\kern-.025em b}\kern-.08em
    T\kern-.1667em\lower.7ex\hbox{E}\kern-.125emX}}
\usepackage{hyperref}
\usepackage{subcaption}
\usepackage{booktabs} 
\usepackage{multicol}

\usepackage{amsthm}

\makeatletter
\renewenvironment{proof}[1][\proofname]{\par
  \pushQED{\qed}%
  \normalfont \topsep6\p@\@plus6\p@\relax
  \trivlist
  \item[\hskip\labelsep
        {\bfseries\itshape
    #1:}\@addpunct{}]\ignorespaces
}{%
  \popQED\endtrivlist\@endpefalse
}
\makeatother

\newtheoremstyle{boldhead}
  {} % Space above
  {} % Space below
  {\normalfont} % Body font - non-italic
  {} % Indent amount
  {\bfseries\itshape} % Theorem head font - bold and italic
  {.} % Punctuation after theorem head
  { } % Space after theorem head
  {\thmname{#1}\thmnumber{ #2}\thmnote{ \normalfont(#3)}} % Theorem head spec (for definitions and remarks)

\theoremstyle{bolditalichead}
\newtheorem{theorem}{Theorem}
\newtheorem{problem}{Problem}
\newtheorem{assumption}{Assumption}
\newtheorem{definition}{Definition}
\newtheorem{remark}{Remark}
\newtheorem{lemma}{Lemma}

\graphicspath{ {./Figs/} }
\usepackage[noend]{algpseudocode}
\makeatletter
\def\BState{\State\hskip-\ALG@thistlm}
\makeatother
\usepackage{cite}
\usepackage{paralist}
\usepackage{color}

\definecolor{blue}{rgb}{0, 0.1, 0.7}
\usepackage{mathrsfs}
\usepackage{multirow}
\usepackage{graphicx}
\usepackage{subcaption}
\usepackage{caption}
\usepackage[font=small,justification=justified,singlelinecheck=false]{caption}

\ifCLASSINFOpdf
 
\else
 
\fi

\usepackage{xcolor}
\definecolor{blue}{rgb}{0,0,0}

\DeclareMathSizes{10}{9.5}{7}{5} % {main font}{math}{script}{scriptscript}

\begin{document}

\title{Lyapunov-based Resilient Secondary Synchronization Strategy of AC Microgrids Under Exponentially Energy-Unbounded FDI Attacks}

\scriptsize % Make authors and affiliations smaller
\author{
\IEEEauthorblockN{1\textsuperscript{st} Mohamadamin Rajabinezhad~\IEEEmembership{Student Member,~IEEE}}
\IEEEauthorblockA{\textit{Dept. of Electrical and Computer Engineering} \\
\textit{University of Connecticut},
Storrs, CT, USA \\
Mohamadamin.rajabinezhad@uconn.edu}
\vspace{-0.35cm} % Adjust space between authors
\and 
\IEEEauthorblockN{2\textsuperscript{nd} Nesa Shams}
\IEEEauthorblockA{\textit{Dept. of Electrical and Computer Engineering} \\
\textit{University of Connecticut},
Storrs, CT, USA\\
sln24004@uconn.edu}
\vspace{-0.35cm} % Adjust space between authors
\and
\IEEEauthorblockN{3\textsuperscript{rd} Asad Ali Khan~\IEEEmembership{Member,~IEEE}}
\IEEEauthorblockA{\textit{Dept. of Electrical and Computer Engineering} \\
\textit{University of Texas at San Antonio}, San Antonio, TX, USA\\
asad.khan@my.utsa.edu}
\vspace{-0.35cm} % Adjust space between authors
\and 
\IEEEauthorblockN{4\textsuperscript{th} Omar A. Beg~\IEEEmembership{Senior Member,~IEEE}}
\IEEEauthorblockA{\textit{Dept. of Electrical and Computer Engineering} \\
\textit{University of Texas Permian Basin}, Odessa, TX, USA\\
beg\_o@utpb.edu} 
\vspace{-0.35cm} % Adjust space between authors
\and 
\IEEEauthorblockN{5\textsuperscript{th} Shan Zuo~\IEEEmembership{Member,~IEEE}}
\IEEEauthorblockA{\textit{Dept. of Electrical and Computer Engineering} \\
\textit{University of Connecticut},
Storrs, CT, USA \\
shan.zuo@uconn.edu}
}

\normalsize % Reset to normal font size for the rest of the document

% Add negative space between the author block and abstract
\IEEEaftertitletext{\vspace{-0.5cm}} % Adjust this value as needed
\linespread{0.95} % Reduce spacing for a smaller appearance
\maketitle

{\color{blue}

\begin{abstract}

This article presents fully distributed Lyapunov-based attack-resilient secondary control strategies for islanded inverter-based AC microgrids, designed to counter a broad spectrum of energy-unbounded False Data Injection (FDI) attacks, including exponential attacks, targeting control input channels. While distributed control improves scalability and reliability, it also increases susceptibility to cyber threats. The proposed strategies, supported by rigorous Lyapunov-based proofs, ensure uniformly ultimately bounded (UUB) convergence for frequency regulation, voltage containment, and power sharing, even under severe cyber attacks. The effectiveness of the proposed approach has been demonstrated through case studies on a modified IEEE 34-bus system, leveraging simulations and real-time Hardware-in-the-Loop experiments with OPAL-RT.

%The enhanced resilient performance of the proposed cyber-physical defense strategies has been verified through comprehensive case studies on a modified IEEE 34-bus test feeder benchmark system incorporating four inverter-based Distributed Energy Resources (DERs).

\end{abstract}
}
\begin{IEEEkeywords}
Distributed resilient
secondary control, FDI unbounded attacks, AC microgrids, Containment control.           
\end{IEEEkeywords}

\section{Introduction}
AC microgrids in islanded mode typically follow a hierarchical control structure with primary, secondary, and tertiary levels. Distributed control at the secondary level enhances reliability, scalability, and communication efficiency \cite{zuo2022resilient}. However, incorporating information and communication technology increases vulnerability to cyber attacks due to limited global situational awareness \cite{zuo2022resilient,liang2016review}. Severe attacks can go undetected in real-time, making cybersecurity crucial, especially given the low frequency stability and scarce defense resources in isolated microgrids \cite{alhelou2023power}. Common attacks like replay, denial-of-service (DoS), and false data injection (FDI) can disrupt sensor readings, control inputs, and communication networks and affecting synchronization. Given that attack-detection methods may struggle against stealthy attackers \cite{liang2016review}, enhancing the self-resilience of large-scale networked microgrids with attack-resilient control protocols is essential. These distributed protocols maintain performance by mitigating disturbances and attacks without detecting compromised components, focusing on local solutions for resilience \cite{xia2024distributed,zuo2020resilient,zhou2023distributed,wang2024secondary,liu2021robust,shi2021observer,liu2023resilient}. Ref \cite{shi2021observer} proposes a resilient control method that improves conventional distributed control by adding compensation terms based on errors between neighboring frequency and active power signals. Ref \cite{liu2021robust} presents a robust and resilient distributed optimal frequency control for AC microgrids by integrating the cyber-physical system with an auxiliary communication network layer. Most AC microgrid studies treat disturbances, faults, or attacks as bounded signals. However, recent research highlights unbounded false data injections, exploiting quantum computing's capabilities to target various components of cybersystems, maximizing damage and posing severe threats to microgrid stability, especially in islanded systems \cite{zuo2020resilient,zhou2023distributed,wang2024secondary,wiebe2023exponential,lakshmi2023quantum}. Traditional defenses mechanism may be insufficient against these complex attacks.

In this paper, we tackle the practical yet challenging problem of cooperative resilient secondary defence strategy in AC microgrids under a wide range of unbounded attacks, including exponential energy-unbounded FDI (EU-FDI) attacks. Unlike prior studies that primarily address either bounded attacks or so-called unbounded attacks with bounded first derivatives \cite{liu2023resilient,zuo2020resilient,zhou2023distributed}, our methodology relaxes these constraints. Unbounded attacks that target the control input and influence the rate of change of controlled variables can induce rapid fluctuations in these variables before they reach their physical saturation limits, potentially destabilizing the system. This underscores the pressing need for robust defense strategies to ensure microgrid stability amidst sophisticated cyber threats, especially in the emerging quantum area. The contributions of this paper are summarized as follows:

% $\bullet$ We propose fully distributed attack-resilient defense strategies for secondary frequency and voltage control in AC microgrids. A compensational signal is designed to counteract unbounded cyber-physical attacks, including exponential FDI attacks, using an adaptively tuned parameter based on neighborhood information. Unlike existing solutions \cite{liu2023resilient, zuo2020resilient, zhou2023distributed}, which address limited unbounded attack signals with bounded first-order derivatives, our strategies extend our previous work on polynomially unbounded attacks \cite{wang2024secondary} to handle a broader range of threats, including exponential FDI attacks enhancing microgrid defense against malicious attacks.
% $\bullet$ A rigorous proof based on Lyapunov stability analysis proves that the proposed resilient control strategy achieves UUB convergence for frequency regulation, voltage containment, and power sharing, even under exponentially unbounded attacks. The defense strategies mitigate the effects of these attacks, and the ultimate bounds of frequency and voltage can be adjusted by tuning the adaptation gains. 
$\bullet$ We propose fully distributed, attack-resilient defense strategies for secondary frequency and voltage control in AC microgrids. Our approach utilizes a compensational signal designed to counteract unbounded cyber-physical attacks, including EU-FDI, through an adaptively tuned parameter based on neighborhood information. Unlike existing solutions \cite{liu2023resilient, zuo2020resilient, zhou2023distributed}, which handle limited unbounded attack signals with bounded first-order derivatives, our strategies expand upon previous work \cite{wang2024secondary} to address a wider range of threats, enhancing microgrid defenses against malicious attacks.

$\bullet$ A rigorous proof based on Lyapunov stability analysis demonstrates that that the proposed cyber-physical resilient secondary control ensures UUB convergence for frequency regulation, voltage containment, and power sharing, even under exponentially unbounded attacks. 
% The fully distributed defense strategies require no global information, enabling scalability and plug-and-play capability, as demonstrated through case studies on a modified IEEE 34-bus system using simulations and real-time HIL experiments with OPAL-RT.

$\bullet$ The proposed defense strategies are fully distributed, requiring no global information, ensuring scalability and plug-and-play capability. Their effectiveness has been demonstrated through case studies on a modified IEEE 34-bus system using simulations and real-time hardware-in-the-loop experiments with OPAL-RT.
% \begin{figure}[!t]
% \centering
% \includegraphics[width=2.5in]{ACpaperFig1.JPG}
% \caption{A networked multi-inverter system under actuator attacks.}
% \label{FIG2}
% \end{figure}
\section{Preliminaries on Sparse Communication Network}

\subsection{Notation and Graph Theory}
The minimum and maximum singular values, \( \sigma_{\min}(\cdot) \) and \( \sigma_{\max}(\cdot) \), denote a matrix's smallest and largest singular values. Sets \( \mathscr{F} = \{1, 2, \dots, N\} \) and \( \mathscr{L} = \{N+1, N+2\} \) represent follower and leader nodes, respectively, with \( \mathbf{1}_N \) as an all-ones vector. Operators \( \otimes \), \( \operatorname{diag}(\cdot) \), and \( \| \cdot \| \) indicate the Kronecker product, block diagonal matrix, and Euclidean norm. A network with \( N \) inverters and two leader nodes is modeled by the digraph \( \mathscr{G} = (\mathcal{V}, \mathcal{E}, \mathcal{A}) \), where leaders issue reference values, followers communicate via adjacency matrix \( \mathcal{A} = [a_{ij}] \), and node interactions are defined by the Laplacian \( \mathcal{L} = \mathcal{D} - \mathcal{A} \). The pinning gain \( g_{ir} \) indicates leader influence on followers, with \( \mathcal{G}_r = \operatorname{diag}(g_{ir}) \) as the pinning matrix.

\section{Cooperative Control of AC Microgrids}
An inverter-based distributed generation (DG) system consists of a Voltage Source Inverter (VSI) along with internal power, voltage, and current controllers designed to oversee and regulate the terminal voltage and operating frequency of the DG. The primary control level is the local control of DGs, which typically employs droop control techniques. These techniques govern the frequency of DGs by adjusting active power and regulate the voltage magnitude by managing reactive power. The primary droop mechanism for the \(i\)th inverter is expressed as follows:
\begin{align}
& \omega_i = \omega_{n_i}-m_{P_i}P_i,
\label{eq1}\\
& v_{odi} = {V_{n_i}} - {n_{Q_i}}{Q_i},
\label{eq2}
\end{align}
where $P_i$ and $Q_i$ are the active and reactive powers, respectively. $\omega _i$ and $v_{odi}$ are the operating angular frequency and the $d$ component of in $abc$ to $dq0$ transform (park transform) of inverter terminal voltage, respectively. ${m_{P_i}}$ and ${n_{Q_i}}$ are $P-\omega$ and $Q-v$ droop coefficients selected per inverters' power ratings. ${\omega _{n_i}}$ and ${V_{n_i}}$ are the setpoints for the primary droop mechanisms fed from the secondary control layer. The secondary control is to restore the operating frequency and terminal voltage magnitude of DGs to the reference frequency and voltage. Standard secondary control functions as an actuator, supplying input control signals to adjust setpoints in decentralized primary control. We differentiate the droop relations in \eqref{eq1} and \eqref{eq2} with respect to time to obtain
% {\small
\begin{align}
\small
& {{\dot \omega }_{n_i}}={{\dot \omega }_i} + {m_{P_i}}{{\dot P}_i}=u_{f_i},
\label{eq3}\\
& {{\dot V}_{n_i}}={{\dot v}_{odi}} + {n_{Q_i}}{{\dot Q}_i}=u_{v_i},
\label{eq4}
\end{align}

where \( u_{f_i} \) and \( u_{v_i} \) are auxiliary control inputs to be designed later. To synchronize each inverter's terminal frequency and maintain voltage within acceptable limits, we adopt a leader-follower containment-based secondary control \cite{zuo2020resilient}. The local
cooperative frequency and voltage control protocols at each
inverter will be designed based on the following relative
information with respect to the neighboring inverters and the
leaders
% \begin{equation}
% \begin{gathered}
%   {u_{{f_i}}} = {c_{f_i}}\left( {\sum\limits_{j \in \mathscr{F}} {a_{ij}\left( {{\omega _j} - {\omega _i}} \right)} } \right. + \sum\limits_{k \in \mathscr{L}} {g_{ik}\left( {{\omega _k} - {\omega _i}} \right)}  \hfill \\
%   \left. {\quad \quad  + \sum\limits_{j \in \mathscr{F}} {a_{ij}\left( {{m_{{P_j}}}{P_j} - {m_{{P_i}}}{P_i}} \right)} } \right), \hfill \\ 
% \end{gathered}
% \label{eq5}
% \end{equation}
% \begin{equation}
% \begin{gathered}
%   {u_{{v_i}}} = {c_{v_i}}\left( {\sum\limits_{j \in \mathscr{F}} {a_{ij}\left( {v _{odj} - v _{odi}} \right)} } \right. + \sum\limits_{k \in \mathscr{L}} {g_{ik}\left( {{v _k} - v _{odi}} \right)}  \hfill \\
%   \left. {\quad \quad  + \sum\limits_{j \in \mathscr{F}} {a_{ij}\left( {{n_{{Q_j}}}{Q_j} - {n_{{Q_i}}}{Q_i}} \right)} } \right), \hfill \\ 
% \end{gathered}
% \label{eq6}
% \end{equation}
% {\small
\begin{align}
\scriptsize
{{\dot \omega }_{{n_i}}} &= {c_{f_i}}\left( {\sum\limits_{j \in \mathscr{F}} {a_{ij}\left( {{\omega _{{n_j}}} - {\omega _{{n_i}}}} \right)}  + \sum\limits_{k \in \mathscr{L}} {g_{ik}\left( {{\omega _{n_k}} - {\omega _{{n_i}}}} \right)} } \right)\label{eq9}
\end{align}
\begin{equation}
\small
{{\dot V}_{{n_i}}} = {c_{v_i}}\left( {\sum\limits_{j \in \mathscr{F}} {a_{ij}\left( {{V_{{n_j}}} - {V_{{n_i}}}} \right)}  + \sum\limits_{k \in \mathscr{L}} {g_{ik}\left( {V_{n_k} - {V_{{n_i}}}} \right)} } \right)
\label{eq11}
\end{equation}

where $c_{f_i}$ and $c_{v_i}$ are constant gains. The setpoints for the primary-level droop control, $\omega_{n_i}$ and ${V_{n_i}}$, are, then, computed from $u_{f_i}$ and $u_{v_i}$ as ${\omega _{n_i}} = \int {u_{f_i}} \operatorname{d} t,{V _{n_i}} = \int {u_{v_i}} \operatorname{d} t.$ where ${\omega _{{n_k}}}={\omega _k} + {m_{{P_i}}}{P_i}$ and $V_{n_k}={v_k + {n_{{Q_i}}}{Q_i}}$. While the control protocols include power-sharing mechanisms, leading to synchronization of the frequency and voltage of each inverter in the steady state. Define ${\Phi _k} = \frac{1}{2}{\mathcal{L}} + \mathcal{G}_k$. Then, the global forms of \eqref{eq9} and \eqref{eq11} are
% {\small
\begin{align}
    {\dot \omega }_n & = {\xi _{{f}}} \equiv - \operatorname{diag} \left( {{c_{{f_i}}}} \right) \sum\limits_{k \in \mathscr{L}} {{\Phi _k}\left( {{\omega _n} - {{\mathbf{1}}_N} \otimes {\omega _{{n_k}}}} \right),} \label{eq12} \\
    \dot V_n & = {\xi _{{v}}} \equiv - {\operatorname{diag} \left( {{c_{v_i}}} \right)} \sum\limits_{k \in \mathscr{L}} {\Phi _k\left( {V_n - {{\mathbf{1}}_N} \otimes {V_{n_k}}} \right),} \label{eq13}
\end{align}

where $\omega_n= {[ {\omega_{n_1}^T,...,\omega_{n_N}^T} ]^T}$ and $V_n= {[ {V_{n_1}^T,...,V_{n_N}^T} ]^T}$. Define the global frequency and voltage containment error vectors as
% {\small
\begin{align}
    {e_f} & = {\omega _n} - {\left( {\sum\limits_{k \in \mathscr{L}} {{\Phi _k}} } \right)^{ - 1}} \sum\limits_{k \in \mathscr{L}} {{\Phi _k}\left( {{{\mathbf{1}}_N} \otimes {\omega _{{n_k}}}} \right)} , \label{eq14} \\
    {e_v} & = V_n - {\left( {\sum\limits_{k \in \mathscr{L}} {{\Phi _k}} } \right)^{ - 1}} \sum\limits_{k \in \mathscr{L}} {{\Phi _k}\left( {{{\mathbf{1}}_N} \otimes V_{n_k}} \right)}. \label{eq15}
\end{align}
The following assumption is needed for the communication graph topology to guarantee cooperative consensus.
\begin{assumption}
\label{ass: leader follower}
There exists a directed path from at least one leader to each inverter.
\end{assumption}

\begin{lemma}[\cite{zuo2020resilient}]
\label{le: control objective} 
Suppose Assumption \ref{ass: leader follower} holds, $\sum\nolimits_{k \in \mathscr{L}} {\Phi _k} $ is non-singular and positive-definite. 
In the absence of attack, using the designed cooperative secondary control \eqref{eq9} and \eqref{eq11}, the frequency and voltage containment control objectives are achieved if $\mathop {\lim }\limits_{t \to \infty } e_f \left( t \right) = 0$ and $\mathop {\lim }\limits_{t \to \infty } e_v \left( t \right) = 0$, respectively.
\end{lemma}
\begin{figure}[t]
\centering
\includegraphics[width=2in]{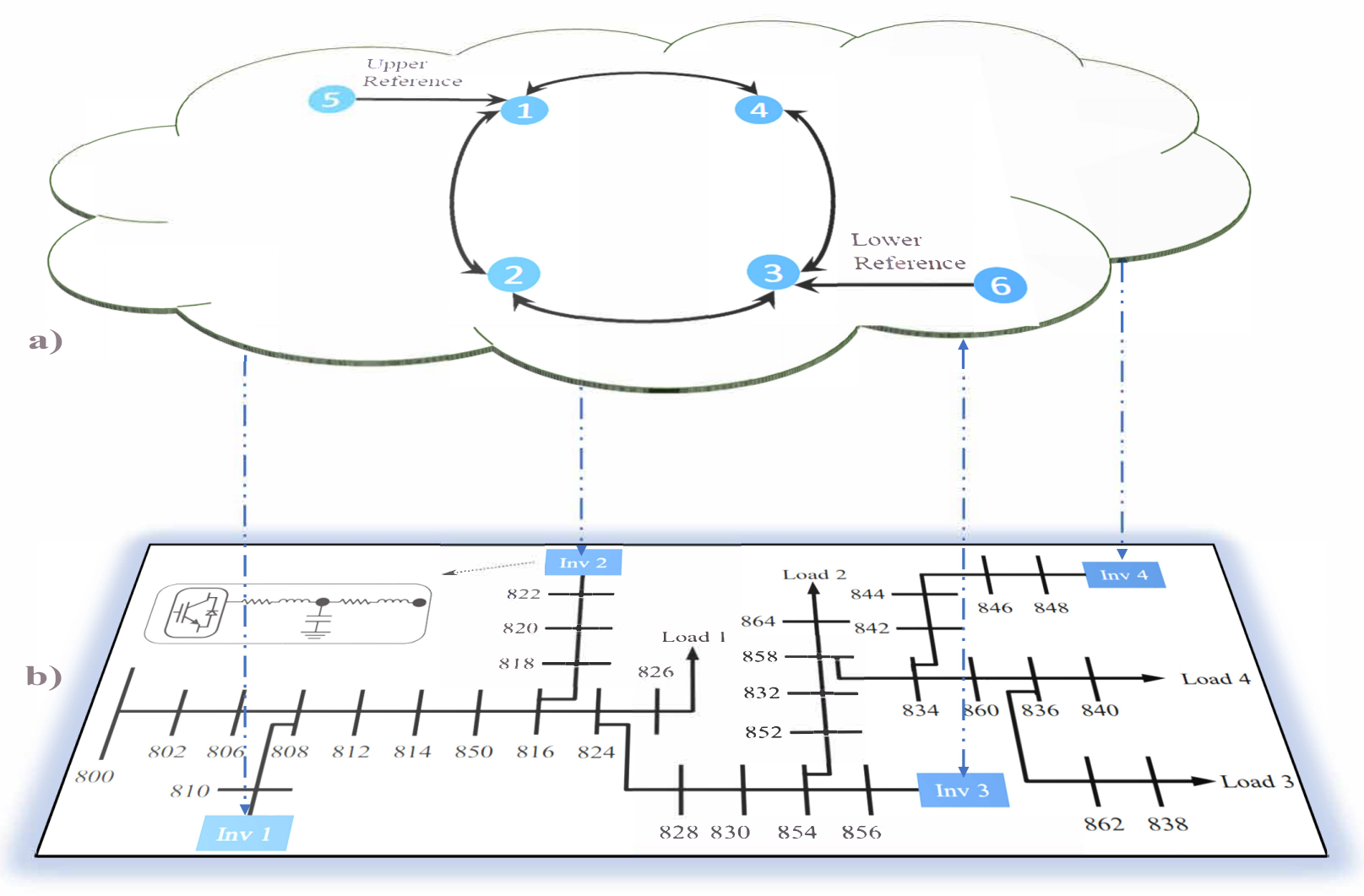}
%{FIG3.pdf}
\caption{Cyber-physical microgrid system: (a) Communication graph topology among four inverters and two leaders (references), (b) IEEE 34-bus system with four inverters.}
\label{FIG123}
\end{figure}
\vspace{-3mm}
\section{Problem Formulation}
In this section, we formulate the resilient defense problems for the secondary frequency and voltage control of an AC microgrid. Specifically, we introduce the EU-FDI attacks on the local control inputs of the frequency and voltage control loops, then auxiliary control input signal in \eqref{eq3} and \eqref{eq4}, becomes to: 
% {\small
\begin{align}
    \bar{u}_{f_i} & = u_{f_i} + \mu_{f_i} \quad \text{,} \quad \bar{u}_{v_i} = u_{v_i} + \mu_{v_i} \label{eq5+} 
\end{align}
where $\bar u_{f_i}$, and $\bar u_{v_i}$ denote the ith component of corrupted control signals received by actuators and $\mu_{f_i}$ and $\mu_{v_i}$ denote the unbounded attack signals injected to the input channels of frequency and voltage control loops at the $i^{th}$ inverter, respectively.
\begin{definition}
\label{def: exponentially unbounded signal}
A signal $\mu(t)$ is said to be exponentially unbounded if its norm grows at most exponentially with time, i.e., $\|\mu(t)\| \leqslant \exp(\kappa t)$, where $\kappa$ is a positive constant.
\end{definition}

\begin{assumption}
\label{ass: attacks}
$\mu_{f_i}(t)$ and $\mu_{v_i}(t)$ are exponentially unbounded signals, i.e.,
$\|{{\mu_{f_i}}} \|\le \gamma_i \exp(\rho_it) $ and $\|{{\mu_{v_i}}} \|\le \gamma_i \exp(\rho_it) $, where $\rho_i$ and $\gamma_i$ are positive constant.
\end{assumption}
\begin{remark}
\label{rem: compromised observer}
% In the secondary control mechanism \cite{zhou2023distributed,zuo2020resilient,liu2023resilient}, the control input $u_i$, which is the rate of change of $V_{n_i}$ over time ($\dot{V}_{n_i}$), is generated in a virtual layer before the saturation mechanism takes effect. Here, the virtual control input $u_i = \dot{V}_{n_i}$, meaning that if a fast-growing unbounded signal is injected into the virtual auxiliary control input, the rate of change of $V_{n_i}$ will also become fast-growing and unbounded over time, leading to potential system instability before the saturation mechanism takes action. The projection of cyberattack signals increasing at an exponential rate reflects a plausible risk in the quantum era, where limitations on the power of attack signals are increasingly less defined. Traditional defenses, designed for bounded disturbances, are inadequate against these types of attacks. Inspired by previous work on unbounded attacks \cite{zuo2020resilient,zhou2022resilient,liu2023resilient}, it is noted that the controller's operation in a virtual layer makes it vulnerable to adversaries injecting fast-growing time-varying signals, causing the control input to vary without bounds, compromising system stability.
In the secondary control mechanism, the control input $u_i = \dot{V}_{n_i}$ is generated in a virtual layer \cite{zhou2023distributed,zuo2020resilient,liu2023resilient}. If an unbounded, fast-growing signal is injected, the rate of change $\dot{V}_{n_i}$ can become uncontrollable before the saturation mechanism activates, leading to system instability. Inspired by previous work on unbounded attacks \cite{zuo2020resilient,zhou2022resilient,liu2023resilient}, this vulnerability is especially concerning in the quantum era, where exponentially increasing attack signals can bypass traditional defenses designed for bounded disturbances.
\end{remark}
Since $\mu_{f_i}$ and $\mu_{v_i}$ are unbounded, conventional cooperative control fails to regulate frequency and contain voltages within acceptable ranges. Attack-resilient strategies are needed to ensure frequency regulation, voltage containment, and closed-loop stability. The following convergence definition applies.
\begin{definition}[\cite{khalil2002control}]
\label{def: UUB}
Signal $x(t)$ is UUB with an ultimate bound $b$, if there exist positive constants $b$ and $c$, independent of ${t_0} \geq 0$, and for every $a \in \left( {0,c} \right)$, there exist $t_1 = t_1 \left( {a,b} \right) \geq 0$, independent of $t_0$, such that $\left\| {x\left( {{t_0}} \right)} \right\| \leq a\Rightarrow \left\| {x\left( t \right)} \right\| \leq b, \forall t \geq {t_0} + {t_1}$.
\end{definition}

% \textbf{\textit{Definition} 1 (\cite{khalil2002control}):} 
% Signal $x(t)$ is UUB with an ultimate bound $b$, if there exist positive constants $b$ and $c$, independent of ${t_0} \geq 0$, and for every $a \in \left( {0,c} \right)$, there exist $t_1 = t_1 \left( {a,b} \right) \geq 0$, independent of $t_0$, such that $\left\| {x\left( {{t_0}} \right)} \right\| \leq a\Rightarrow \left\| {x\left( t \right)} \right\| \leq b, \forall t \geq {t_0} + {t_1}$.

Now, we introduce the following attack-resilient defense problems for the secondary frequency and voltage control loops.
\begin{problem}[Attack-resilient Frequency Defense Problem]
\label{pro: Attack-resilient Frequency Defense Problem}
The aim is to design an input control signal $u_{f_i}$, as delineated in Eq. \eqref{eq3}, for each inverter, such that the global frequency containment error $e_f$, as specified in Eq. \eqref{eq14}, remains UUB in the face of broad range of unbounded attacks including EU-FDI attacks on the local frequency control loop.
\end{problem}

% \textbf{\textit{Definition} 2 \textit{(Attack-resilient Frequency Defense Problem)}:}
% The aim is to design an input control signal $u_{f_i}$, as delineated in equation \eqref{eq3}, for each inverter, such that the global frequency containment error $e_f$, as specified in equation \eqref{eq14}, remains UUB in the face of generally unbounded attacks on the local frequency control loop. 

\begin{problem}[Attack-resilient Voltage Defense Problem]
\label{pro: Attack-resilient Voltage Defense Problem}
The aim is to design an input control signal $u_{v_i}$, as delineated in Eq. \eqref{eq4}, for each inverter, such that the global voltage containment error $e_v$, as defined in Eq. \eqref{eq15}, remains UUB in the face of broad range of unbounded attacks including EU-FDI attacks on the local voltage control loop.
\end{problem}
% \textbf{\textit{Definition} 3 \textit{(Attack-resilient Voltage Defense Problem)}:}
% The aim is to design an input control signal $u_{v_i}$, as delineated in equation \eqref{eq4}, for each inverter, such that the global voltage containment error $e_v$, as defined in equation \eqref{eq15}, remains UUB in the face of generally unbounded attacks on the local voltage control loop. 
\section{Fully Distributed Attack-resilient Defense Strategies Design and Stability Analysis}
% \begin{figure}[!t]
% \centering
% \includegraphics[width=3.5in]{Controller.pdf}
% \caption{Communication layer among inverters, and the proposed attack-resilient secondary defense strategies for an inverter.}
% \label{FIG2}
% \end{figure}
We propose the following fully distributed attack-resilient defense strategies to solve the attack-resilient frequency and voltage defense problems.
% \begin{equation}
% \left\{ \begin{gathered}
%   u_{f_i} = {\xi _{f_i} + {\Gamma}_{f_i}}, \hfill \\
%   {{\Gamma}_{f_i}} = \frac{\xi_{f_i}e^{\varphi_{f_i}}}{\lvert\xi_{f_i}\rvert + \eta_{f_i}}, \hfill \\ 
%   {\dot\varphi_{f_i}} = \beta_{f_i}\Bigg(\lvert\xi_{f_i}\rvert - \upsilon_{fi}\bigg(\varphi_{f_i}-\hat\varphi_{f_i}\bigg)\Bigg), \hfill \\ 
%   \dot{\hat{\varphi}}_{f_i} = \kappa_{f_i}\bigg(\varphi_{f_i}-\hat\varphi_{f_i}\bigg), \hfill
% \end{gathered}  \right.    
% \label{eq20} 
% \end{equation}
% \begin{equation}
% \left\{ \begin{gathered}
%   u_{v_i} = {\xi _{v_i} + {\Gamma}_{v_i}}, \hfill \\
%   {{\Gamma}_{v_i}} = \frac{\xi_{v_i}e^{\varphi_{v_i}}}{\lvert\xi_{v_i}\rvert + \eta_{v_i}}, \hfill \\ 
% {\dot\varphi_{v_i}} = \beta_{v_i}\Bigg(\lvert\xi_{v_i}\rvert - \upsilon_{vi}\bigg(\varphi_{v_i}-\hat\varphi_{v_i}\bigg)\Bigg), \hfill \\ 
%   \dot{\hat{\varphi}}_{v_i} = \kappa_{v_i}\bigg(\varphi_{v_i}-\hat\varphi_{v_i}\bigg), \hfill
% \end{gathered}  \right.
% \label{eq21}
% \end{equation}
% {\small
\begin{equation}
\small
\left\{ \begin{aligned}
u_{f_i} &= \xi_{f_i} + \Gamma_{f_i} \\
\Gamma_{f_i} &= \frac{\xi_{f_i}e^{\varphi_{f_i}}}{|\xi_{f_i}| + \eta_{f_i}} \\
\dot{\varphi}_{f_i} &= \beta_{f_i}\left(|\xi_{f_i}| - \lambda_{fi} \right) \\ \lambda_{fi} &= \upsilon_{fi}(\varphi_{f_i}-\hat{\varphi}_{f_i})\\
\dot{\hat{\varphi}}_{f_i} &= \kappa_{f_i}(\varphi_{f_i}-\hat{\varphi}_{f_i})
\end{aligned} \right.
\left\{ \begin{aligned}
u_{v_i} &= \xi_{v_i} + \Gamma_{v_i} \\
\Gamma_{v_i} &= \frac{\xi_{v_i}e^{\varphi_{v_i}}}{|\xi_{v_i}| + \eta_{v_i}} \\
\dot{\varphi}_{v_i} &= \beta_{v_i}\left(|\xi_{v_i}| - \lambda_{fi} \right) \\ \lambda_{vi} &= \upsilon_{vi}(\varphi_{v_i}-\hat{\varphi}_{v_i}) \\
\dot{\hat{\varphi}}_{v_i} &= \kappa_{v_i}(\varphi_{v_i}-\hat{\varphi}_{v_i})
\end{aligned} \right. \label{eq20, eq21}
\end{equation}
where $\eta_{f_i}$ and $\eta_{v_i}$ are positive exponentially decaying functions, ${\Gamma}_{f_i}$ and ${\Gamma}_{v_i}$ are compensational signals, $\varphi_{f_i}$ and $\varphi_{v_i}$ are adaptively tuned parameters, the adaptation gains $\beta_{f_i}$ and $\beta_{v_i}$ are given positive constants. The initial values of both $\varphi_{f_i}$ and $\varphi_{v_i}$ are positive. 
\begin{theorem}
\label{thm: solve pro: Attack-resilient Frequency Defense Problem} 
Under Assumptions \ref{ass: leader follower}, and \ref{ass: attacks} given the implementation of the cooperative attack-resilient frequency defense strategies as delineated in equations \eqref{eq12} and \eqref{eq20, eq21}, the error $e_f$, defined in Eq. \eqref{eq14}, is UUB, i.e., the attack-resilient frequency defense problem is solved. Additionally, it is observed that by properly adjusting the value of ${\beta _{f_i}}$ as prescribed in Eq. \eqref{eq20, eq21}, the ultimate bound of $e_f$ is reduced to an arbitrarily small value.
\end{theorem}
\begin{proof}
See Appendix A for the proof of Theorem~\ref{thm: solve pro: Attack-resilient Frequency Defense Problem}.
\end{proof}
% \textbf{\textit{Theorem} 1:} Under Assumptions 1 and 2, and given the implementation of the cooperative attack-resilient voltage defense strategies as delineated in equations \eqref{eq18} and \eqref{eq20}, the error $e_f$, defined in equation \eqref{eq14}, is UUB, i.e., the attack-resilient frequency defense problem is solved. Additionally, it is observed that by properly adjusting the value of ${\beta _{f_i}}$ as prescribed in equation \eqref{eq20}, the ultimate bound of $e_f$ is reduced to an arbitrarily small value.
\begin{theorem}
\label{thm: solve pro: Attack-resilient Voltage Defense Problem}
Under Assumptions \ref{ass: leader follower}, and \ref{ass: attacks}, the cooperative attack-resilient voltage defense strategies described by Eqs. \eqref{eq13} and \eqref{eq20, eq21} ensure that $e_v$ in Eq. \eqref{eq15} is UUB, i.e., the attack-resilient voltage defense problem is solved. Moreover, by properly adjusting the adaptation gain ${\beta_{v_i}}$ in Eq. \eqref{eq20, eq21} the ultimate bound of $e_v$ is set arbitrarily small.
\end{theorem}
\begin{proof} 
The approach used to prove Theorem \ref{thm: solve pro: Attack-resilient Voltage Defense Problem} mirrors that of Theorem \ref{thm: solve pro: Attack-resilient Frequency Defense Problem}. 
\end{proof}
% \hfill\(\blacksquare\)

% \textbf{\textit{Remark} 2:} 
% Compared to \cite{zuo2020resilient}, the proposed control protocols \eqref{eq18}-\eqref{eq21} have the following merits: (i) Local observers with additional communication information flow were constructed in \cite{zuo2020resilient} to estimate the actual state measurements. This, however, could introduce additional computational complexity. Moreover, the additional communication channels for exchanging observer states could potentially increase the system vulnerability to malicious cyber attacks; (ii) While both \cite{zuo2020resilient} and this letter preserve the UUB convergences for both frequency and voltage terms, in this letter, the ultimate bound can be reduced by properly increasing the adaptive tuning parameters. 
\vspace{-3mm}
\section{Case Studies}
\subsection{Simulation Results}

% {\color{red}
% {directly from \cite{lu2023concurrent}}
% To validate the proposed fully distributed attack-resilient secondary defense strategy in the isolated MG, real-time simulations of an IEEE 34-bus feeder system, islanded at bus 800 and incorporating four inverter-based DERs and two leaders (references), as shown in Fig. \ref{FIG3}. Modulation emulations are included in tests and implemented in the real-time simulation environment Opal-RT.
% }
%The proposed fully distributed attack-resilient secondary defense strategies are validated on an IEEE 34-bus feeder system, islanded at bus 800 and incorporating four inverter-based DERs and two leaders (references), as shown in Fig. \ref{FIG3}. 
The proposed distributed resilient control method is implemented on an IEEE 34-bus balanced test feeder upgraded with
four inverters, as illustrated in Figures \ref{FIG123}. This section presents two different case studies to show the effectiveness of the proposed resilient secondary
synchronization Strategy. Specifications of inverters and their grid-interconnections are adopted from \cite{bidram2014distributed}. All inverters have the same power ratings. The inverter droop gains are set as $m_{P_1}=m_{P_2}=9.4\times {10^{ - 5}}$, $m_{P_3}=m_{P_4}=18.8\times {10^{ - 5}}$, $n_{Q_1}=n_{Q_2}=1.3\times {10^{ - 3}}$, and $n_{Q_3}=n_{Q_4}=2.6\times {10^{ - 3}}$. The inverters communicate on a bidirectional communication network with the adjacency matrix of $\mathcal{A}=[0~1~0~1;1~0~1~0;0~1~0~1;1~0~1~0]$. The pinning gains are $g_{15}=g_{36}=1$. The frequency reference, upper voltage reference, and lower voltage reference are $60 \,\operatorname{Hz}$,  $350\,\operatorname{V}$, and  $330\,\operatorname{V}$, respectively. The performance of the resilient defense strategies defined in \eqref{eq20, eq21} is compared to the conventional secondary control method in \eqref{eq9} and \eqref{eq11}. For the conventional control, the gains are set as \(c_{f_i}=20\) and \(c_{v_i}=10\) for \(i=1,2,3,4\). The adaptation gains for the resilient strategies are \(\beta_{v_i}=20\) and \(\beta_{f_i}=350\). The parameters \(\eta_{v_i}\) and \(\eta_{f_i}\) are defined as \(e^{-\alpha_{v_i}}\) and \(e^{-\alpha_{f_i}}\), with \(\alpha_{v_i}=\alpha_{f_i}=0.01\). To demonstrate the effectiveness of the resilient secondary defense controllers against wide range of unbounded FDI attacks, we use the attack scenarios in Table \ref{table1}. For the conventional secondary control case study, we just considered the last column of unbounded attacks. Fig. \ref{FIG50} compare the voltage and frequency responses to these attacks for both strategies. Results show that, under the conventional approach, voltage and frequency diverge after the attack at \(t=5 \,\mathrm{s}\), leading to instability and improper power sharing. In contrast, the proposed resilient strategies stabilize voltages within \(330–350\,\mathrm{V}\), maintain frequency at \(60\,\mathrm{Hz}\), and ensure equal power sharing after transient fluctuations, despite various FDI attacks. These strategies achieve UUB convergence for frequency regulation, maintain voltage containment, and ensure stable operation of multi-inverter AC microgrids, even under a broad range of unbounded attacks, including EU-FDI attacks.
\begin{figure}[ht]
    \centering
    \begin{subfigure}[b]{0.22\textwidth} % Reduced width
        \centering
        \includegraphics[height=4.55cm]{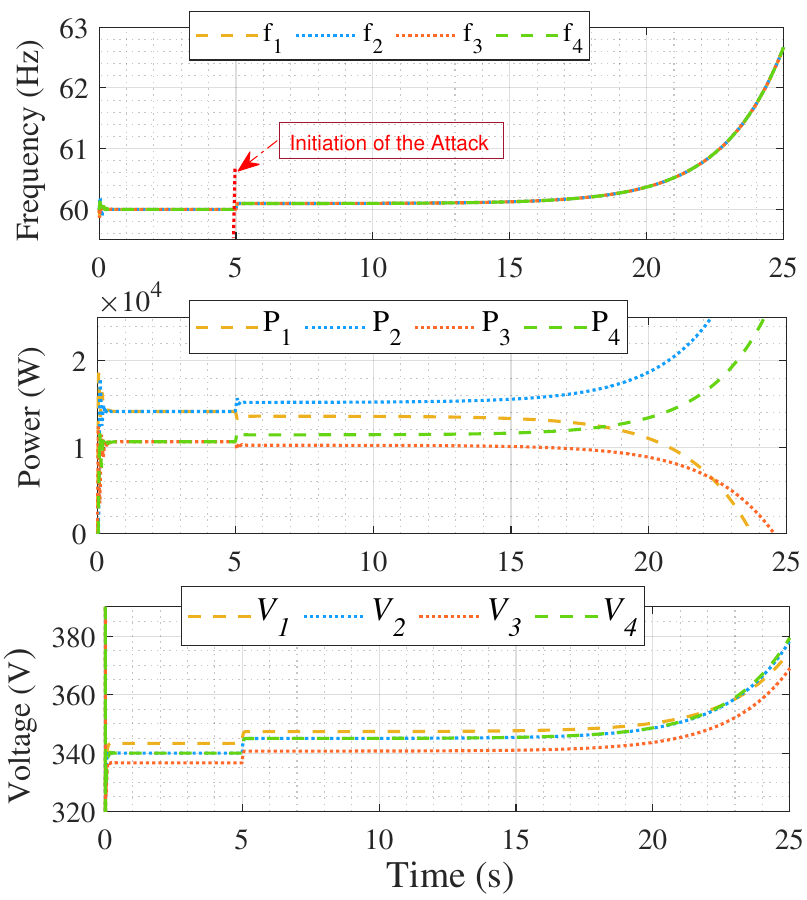}
        % \caption{}
        \label{fig:conventional}
    \end{subfigure}
    \hspace{0.01\textwidth} % Adjust space between figures as needed
    \begin{subfigure}[b]{0.22\textwidth} % Reduced width
        \centering
        \includegraphics[height=4.55cm]{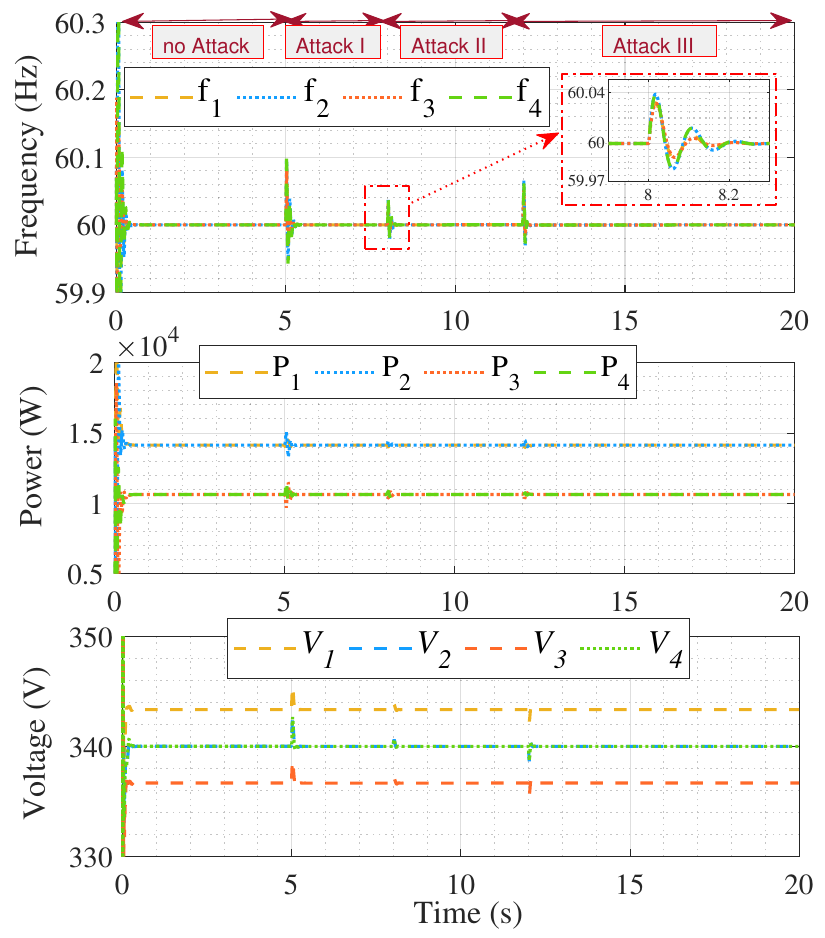}
        % \caption{}
        \label{fig:proposed}
    \end{subfigure}
    \caption{Comparative performance of the (left) conventional and (right) proposed attack-resilient control strategies under unbounded attacks.}
    \label{FIG50}
\end{figure}
\vspace{-5mm}
\subsection{Expremental Validation}
% To validate the efficacy of our proposed resilient control strategy, a AC microgrid system model composed of four DGs is constructed in OPAL-RT 5650. The operating time of the system is $t= 10 s$. The primary control and secondary control are implemented at $t = 0 s$. All the parameters of the controller are the same for experimental validation as well. We have considered the exponentially FDI attacks initiated at $t= 5s$. Upon the injection of the attack, the AC bus frequency and the AC main bus voltage quickly return to their reference values within 0.5 seconds, following minor fluctuations, as shown in Figs.~\ref{Opal_results}(a) and Figs.~\ref{Opal_results}(b). Additionally, the proportional allocation of DGs' active power is achieved within less than 0.5 seconds, as illustrated in Fig.~\ref{Opal_results}(c).
To validate the proposed resilient control strategy, an AC microgrid model with four DGs is constructed in OPAL-RT 5650. The system operates for \( t = 10 \, \text{s} \), with primary and secondary control initiated at \( t = 0 \, \text{s} \). Exponentially unbounded FDI attacks begin at \( t = 5 \, \text{s} \). After the attack, the AC bus frequency and main bus voltage quickly return to their reference values, with minor fluctuations as shown in Figs.~\ref{Opal_results}(a) and (b). Proportional allocation of DGs' active power is also achieved within 0.5 seconds, as shown in Fig.~\ref{Opal_results}(c).
% \begin{figure}
% \centering
% {\includegraphics[width=2in]{opal_structure}}
% \caption{Experimental tested system}
% \label{FIG_stru_Opal}
% \end{figure}
% \vspace{-10mm}
\begin{figure}[!t]
\centering
% Top figure
\begin{subfigure}[b]{0.4\textwidth}
    \centering
    \includegraphics[width=1.6in]{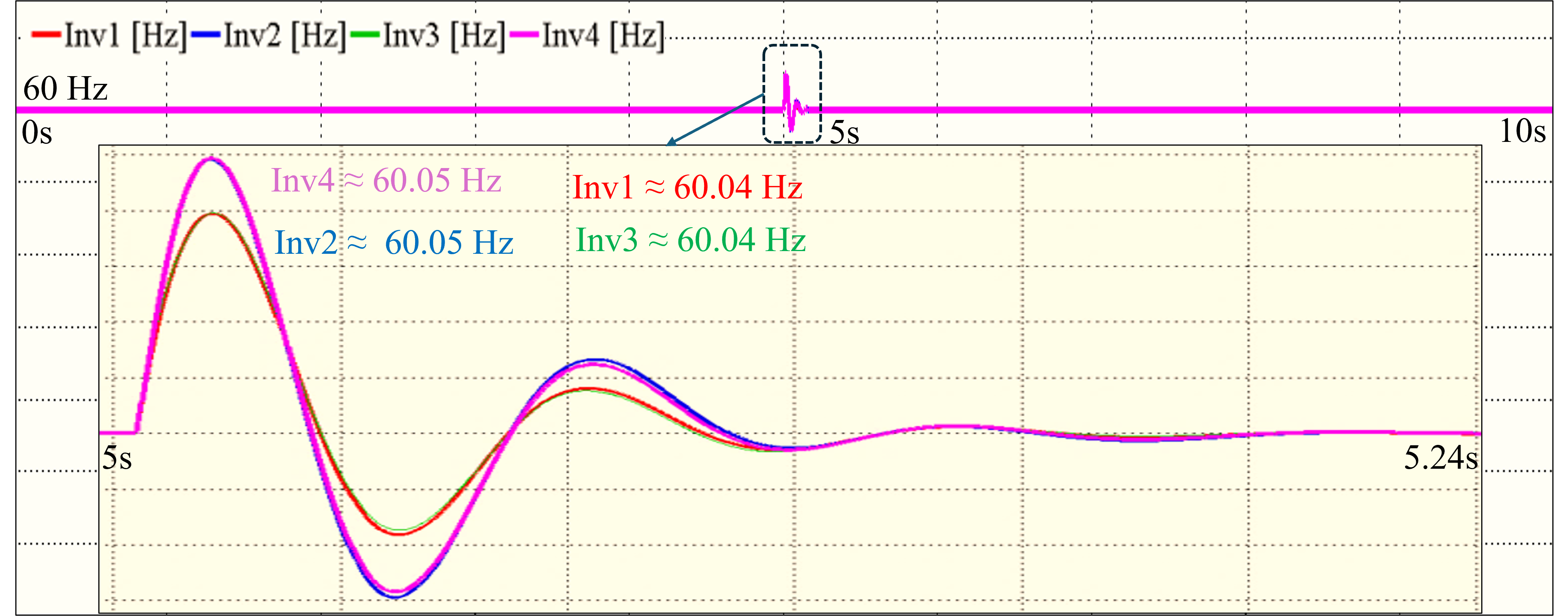}
    \caption{}
    \label{FIGF_Opal}
\end{subfigure}
\vspace{0.1cm} % Adjust the vertical spacing as needed

% Bottom two figures side by side
\begin{subfigure}[b]{0.23\textwidth}
    \centering
    \includegraphics[width=1.6in]{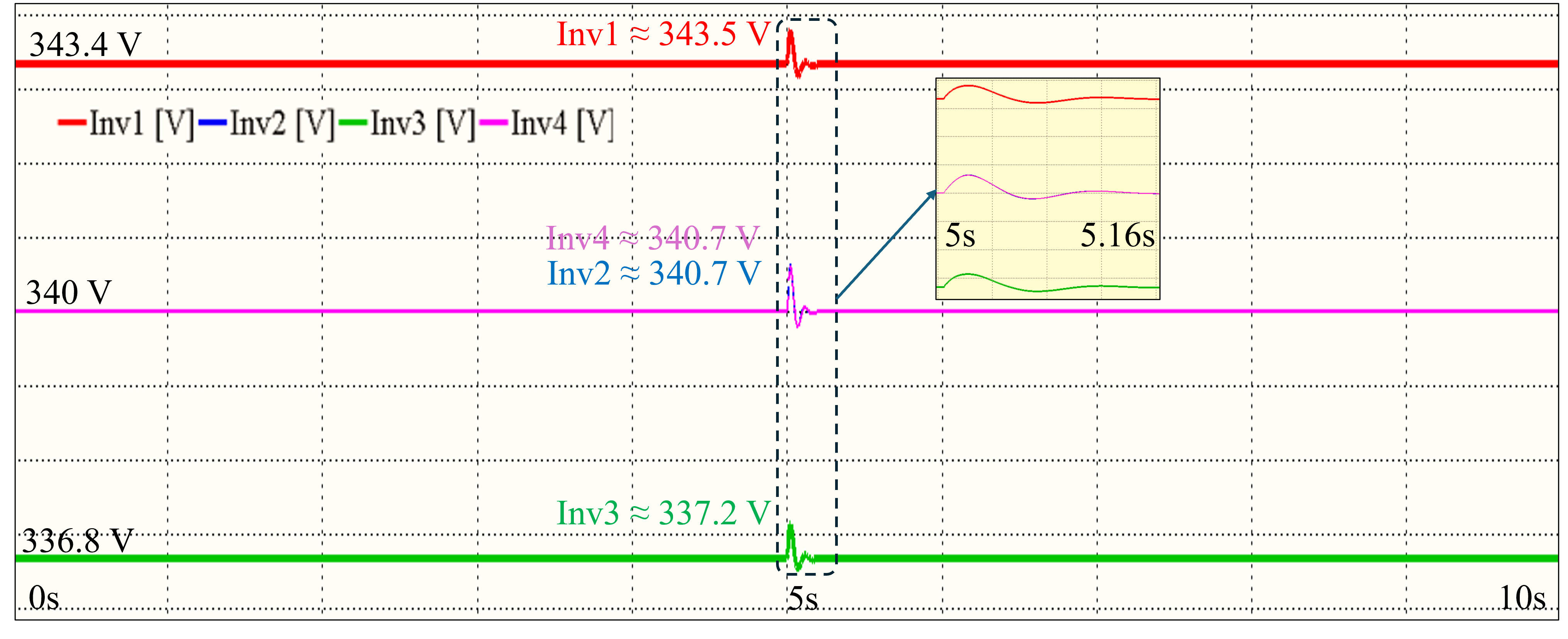}
    \caption{}
    \label{FIGP_Opal}
\end{subfigure}
\hfill
\begin{subfigure}[b]{0.23\textwidth}
    \centering
    \includegraphics[width=1.71in]{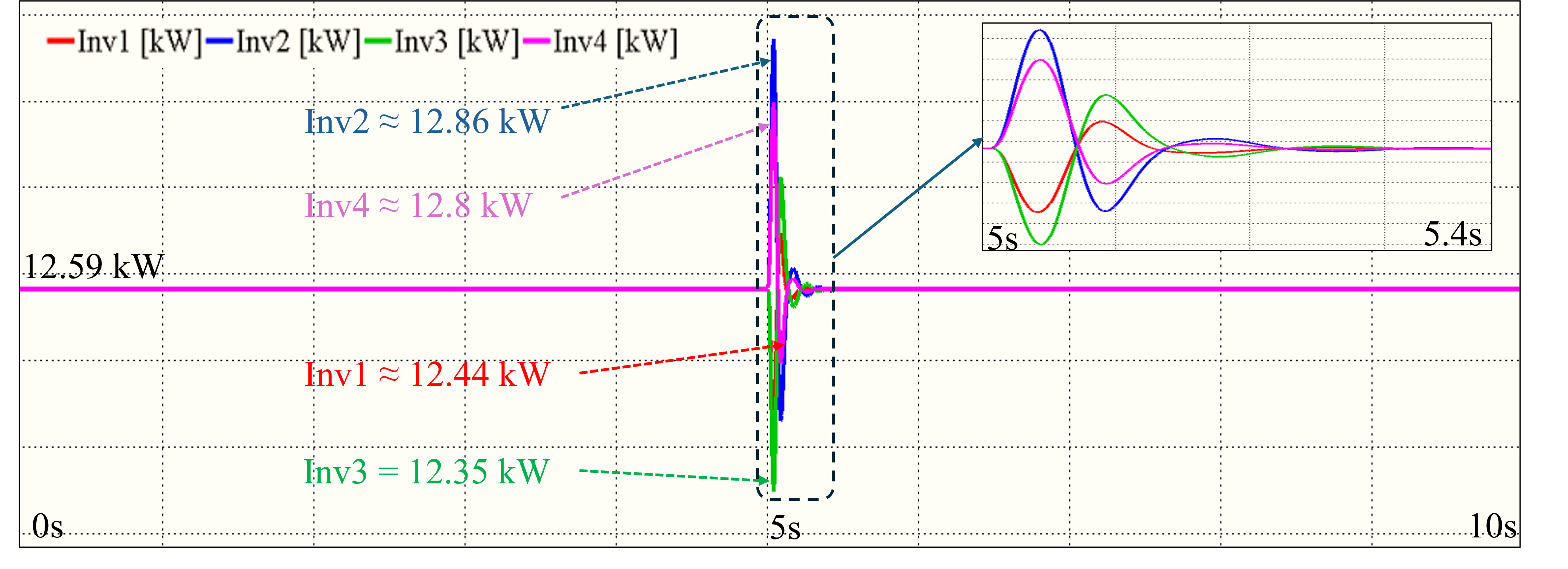}
    \caption{}
    \label{FIGV_Opal}
\end{subfigure}

\caption{Experimental results. (a) Frequency performance
 ,(b) Voltage Performance, (c) Active Power Sharing.}
\label{Opal_results}
\end{figure}
% In general, the resilient control method ensures stable operation of the AC microgrid with four DGs, even under simultaneous exponentially unbounded FDI attacks. Fast regulation and minimized oscillations are achieved for AC bus frequency, voltage, and DGs' active power.
The resilient control method ensures stable operation of the AC microgrid with four DGs, even under simultaneous exponentially unbounded FDI attacks, achieving fast regulation and minimal oscillations for AC bus frequency, voltage, and DGs' active power.
\vspace{-1.5mm}
\section{Conclusion}
This paper has presented novel secondary cyber-physical defense strategies for multi-inverter AC microgrids against broad range of unbounded attacks including exponentially unbounded attacks, on input channels of both frequency and voltage control loops. The proposed fully distributed cyber-physical defense strategies based on adaptive control techniques ensure the UUB stability of the closed-loop system by preserving the UUB consensus for frequency regulation and achieving voltage containment. Moreover, the ultimate bounds of convergence can be tuned by properly adjusting the adaptation gains, $\beta_{f_i}$ and $\beta_{v_i}$, in the adaptive tuning laws. The enhanced resilient performance of the proposed cyber-physical defense strategies has been verified using a modified IEEE 34-bus system. Finally, the effectiveness of the designed resilient distributed secondary control method is validated through simulation and real-time controller hardware-in-the-loop experiment using OPAL-RT.
\vspace{-3mm}
\begin{appendices}
\label{appendix: proofs}
\section{Proof}
\begin{proof}
Combining \eqref{eq9}, \eqref{eq12}, \eqref{eq5+}, and \eqref{eq20, eq21} yields the global form: 
\begin{align}
\small
\dot{\xi}_f& = - \left(\sum_{k \in \mathscr{L}} \Phi_k\right)\operatorname{diag}\left(c_{f_i}\right)\big(\xi_f +\mu_f +\Gamma_f\big),
\label{eq22}
\end{align}
where $\xi_f= [ \xi_{f_i}^T,...,\xi_{f_N}^T ]^T, \mu_f= [ \mu_{f_i}^T,...,\mu_{f_N}^T ]^T$ and $\Gamma_f= [ \Gamma_{f_i}^T,...,\Gamma_{f_N}^T ]^T$. Consider the following Lyapunov function candidate as \ref{eq23}. So its time derivative is \ref{eq24}.
% {\small
\begin{align}
\scriptsize
E & =\frac{1}{2} \xi_f^{T}\left(\sum_{k \in \mathscr{L}} \Phi_k\right)^{-1} \xi_f.
\label{eq23}
\end{align}
\vspace{-4.5mm}
% {\small
\begin{align}
\scriptsize
\dot{E} & =\frac{1}{2} \times 2 \xi_f^{T}\left(\sum_{k \in \mathscr{L}} \Phi_k\right)^{-1} \dot{\xi}_f \nonumber\\
& =-\xi_f^{T}\left(\sum_{k \in \mathscr{L}} \Phi_k\right)^{-1}\left(\sum_{k \in \mathscr{L}} \Phi_k\right) \operatorname{diag}\left(c_{f_i}\right)\big(\xi_f+\mu_f +\Gamma_f\big) \nonumber\\
& \leqslant-\sigma_{\min }\left(\operatorname{diag}\left(c_{f_i}\right)\right)\left\|\xi_f\right\|^2-\operatorname{diag}\left(c_{f_i}\right)\sum\limits_{i \in \mathscr{F}} \big(\xi_{f_i} \mu_{f_i}\big)\nonumber\\
&-\operatorname{diag}\left(c_{f_i}\right)\sum\limits_{i \in \mathscr{F}} \big(\xi_{f_i} \Gamma_{f_i}\big)
\nonumber\\
& \leqslant -\sigma_{\min }\left(\operatorname{diag}\left(c_{f_i}\right)\right)\left\|\xi_f\right\|^2 + \operatorname{diag}\left(c_{f_i}\right)\sum\limits_{i \in \mathscr{F}} \big|\xi_{f_i} \big|\big|\mu_{f_i}\big|\nonumber\\
&-\operatorname{diag}\left(c_{f_i}\right)\sum\limits_{i \in \mathscr{F}} \big(\xi_{f_i} \Gamma_{f_i}\big).
\label{eq24}
\end{align}
Upon substituting $\Gamma_{f_i}$ from Eq. \eqref{eq20, eq21}, the final two terms on the right-hand side of Eq. \eqref{eq24} are transformed as follows
% {\small
\begin{align}
\small
& \operatorname{diag}\left(c_{f_i}\right)\sum\limits_{i \in \mathscr{F}} \big|\xi_{f_i} \big|\big|\mu_{f_i}\big|-\operatorname{diag}\left(c_{f_i}\right)\sum\limits_{i \in \mathscr{F}} \big(\xi_{f_i} \Gamma_{f_i}\big) \nonumber\\
& =  \operatorname{diag}\left(c_{f_i}\right)\sum\limits_{i \in \mathscr{F}}\left(\big|\xi_{f_i}\big| \frac{\big|\xi_{f_i}\big|\big(\big|\mu_{f_i}\big| - e^{\varphi_{f_i}}\big) + \big|\mu_{f_i}\big|\eta_{f_i}}{\left|\xi_{f_i}\right|+\eta_{f_i}}\right)
\label{eq25}
\end{align}
\normalsize

Since $\eta_{f_i}=e^{-\alpha_{fi}t^2}$ is an exponentially decaying function, based on Assumption 2, $\lim_{{t \to \infty}}\big|\mu_{f_i}\big|\eta_{f_i} = 0$. 
To further simplify the mentioned inequality, from \eqref{eq20, eq21}, and if we have:
% {\small
\begin{align}
\small
& \varphi_{f_i} \geqslant \ln(\left|\mu_{f_i}\right|) \Rightarrow \lvert\xi_{f_i}\rvert - \beta_{f_i}\bigg(\varphi_{f_i}-\hat\varphi_{f_i}\bigg) \geqslant \frac{\frac{\operatorname{d}}{\operatorname{d}t}(\left|\mu_{f_i}\right|)}{\left|\mu_{f_i}\right|}
\label{eq2500}
\end{align}
Define $\tilde{\varphi}_{f_i}\left(t\right)=\varphi_{f_i}\left(t\right)-\hat\varphi_{f_i}\left(t\right)$, so the derivative of $\tilde{\varphi}_{f_i}\left(t\right)$ is 
\begin{equation}
\small
\begin{array}{l}
\dot{\tilde{\varphi}}_{f_i}\left(t\right)=\beta_{f_i}\Bigg(\lvert\xi_{f_i}\rvert - \upsilon_{fi}\bigg(\varphi_{f_i}-\hat\varphi_{f_i}\bigg)\Bigg)-\kappa_{f_i}\bigg(\varphi_{f_i}-\hat\varphi_{f_i}\bigg)\\
\quad\;\;\,\,=\beta_{fi}\lvert\xi_{f_i}\rvert-\left(\beta_{fi}\upsilon_i+\kappa_{f_i}\right)\tilde{\varphi}_{f_i}\left(t\right).
\end{array} 
\label{eq1400}
\end{equation}
The solution of \eqref{eq1400} can be written as
% {\small
\begin{equation}
\begin{split}
\small
\tilde{\varphi}_{f_i}\left(t\right)=&e^{-\left(\beta_{fi}\upsilon_i+\kappa_{f_i}\right)t}\tilde{\varphi}_{f_i}\left(0\right)\\
&+\alpha_{fi}\int_0^te^{-\left(\beta_{fi}\upsilon_i+\kappa_{f_i}\right)\left(t-\tau\right)}\lvert\xi_{f_i}\left(\tau\right)\rvert\operatorname{d}\tau.
\end{split}
\label{eq1500}
\end{equation}
Actually, $\tilde{\varphi}_{f_i}\left(t\right)$ will be UUB. This fact can be proved by considering the following two cases:
1) If $\alpha_{fi}\int_0^te^{-\left(\beta_{fi}\upsilon_i+\kappa_{f_i}\right)\left(t-\tau\right)}\lvert\xi_{f_i}\left(\tau\right)\rvert\operatorname{d}\tau$ is bounded, then clearly $\tilde{\varphi}_{f_i}\left(t\right)$ will be UUB as $t \rightarrow \infty$.
2) If $\lim\limits_{t \to \infty}\alpha_{fi}\int_0^te^{-\left(\beta_{fi}\upsilon_i+\kappa_{f_i}\right)\left(t-\tau\right)}\lvert\xi_{f_i}\left(\tau\right)\rvert\operatorname{d}\tau=\infty$, we can rewrite \eqref{eq1500} as follows:
% {\small
\begin{equation}
\small
\begin{split}
\tilde{\varphi}_{f_i}\left(t\right)=&e^{-\left(\beta_{fi}\upsilon_i+\kappa_{f_i}\right)t}\Bigg(\tilde{\varphi}_{f_i}\left(0\right)\\
&+\alpha_{fi}\int_0^te^{\left(\beta_{fi}\upsilon_i+\kappa_{f_i}\right)\left(\tau\right)}\lvert\xi_{f_i}\left(\tau\right)\rvert\operatorname{d}\tau \Bigg).
\end{split}
\end{equation} 
From L'Hôpital's rule, we can write:
\begin{equation}
\small
\begin{split}
&\lim_{t\to\infty}\frac{\int_0^te^{\left(\beta_{fi}\upsilon_i+\kappa_{f_i}\right)\left(\tau\right)}\lvert\xi_{f_i}\left(\tau\right)\rvert\operatorname{d}\tau}{e^{\left(\beta_{fi}\upsilon_i+\kappa_{f_i}\right)\left(t\right)}}\\
&=\lim_{t\to\infty}\frac{e^{\left(\beta_{fi}\upsilon_i+\kappa_{f_i}\right)\left(t\right)}\lvert\xi_{f_i}\left(t\right)\rvert}{\left(\beta_{fi}\upsilon_i+\kappa_{f_i}\right){e^{\left(\beta_{fi}\upsilon_i+\kappa_{f_i}\right)\left(t\right)}}}  = \lim_{t\to\infty}\frac{\lvert\xi_{f_i}\left(t\right)\rvert}{\left(\beta_{fi}\upsilon_i+\kappa_{f_i}\right)}
\end{split}
\end{equation}
Since $\lim_{t\to\infty}\lvert\xi_{f_i}\left(t\right)\rvert$ is UUB, we obtain that $\tilde{\varphi}_{f_i}\left(t\right)$ is also UUB. According to Definition \ref{def: UUB}, let the ultimate bound of $\tilde{\varphi}_{f_i}\left(t\right)$ to be $\psi$. 
Note that the initial values of the gains are chosen such that $\tilde{\varphi}_{f_i}\left(0\right) \ge0$. As a result, we can continue \eqref{eq2500} as follows:
\begin{align}
\scriptsize
& \lvert\xi_{f_i}\rvert - \beta_{f_i}\psi \geqslant \frac{\frac{\operatorname{d}}{\operatorname{d}t}(\left|\mu_{f_i}\right|)}{\left|\mu_{f_i}\right|}
\label{eq25001}
\end{align}
Hence based on Assumption 1, $\left|{{\mu_{f_i}}} \right|\le \gamma_i e^{\rho_it} $ and \eqref{eq25001}, pick $\lvert\xi_{f_i}\rvert \geqslant \gamma_i \rho_i +\beta_{f_i}\psi$ 
i.e., if we have \eqref{eq2500}, such that $e^{\varphi_{f_i}} \geqslant \left|\mu_{f_i}\right|$.
This suggests that $\exists t_2 > t_1$ such that 
\begin{align}
\scriptsize
&\operatorname{diag}\left(c_{f_i}\right)\sum\limits_{i \in \mathscr{F}} \big|\xi_{f_i} \big|\big|\mu_{f_i}\big|-\operatorname{diag}\left(c_{f_i}\right)\sum\limits_{i \in \mathscr{F}} \big(\xi_{f_i} \Gamma_{f_i}\big) \leqslant 0,\forall t \geqslant t_2.
\label{eq26}
\end{align}
Considering \eqref{eq24}, \eqref{eq25} and above equation yields
\begin{align}
\small
\dot E \leqslant 0,\; 
\forall \left|\xi_{f_i}\right| \geqslant\rho_i +\beta_{f_i}\psi ,\forall t\geqslant t_2.
\label{eq27}
\end{align}
Hence, $\xi_f$ is UUB. From Theorem 4.18 of \cite{khalil2002control}, while the system stability is maintained, the larger the value of the adaptation gain $\beta_{f_i}$, the smaller the ultimate bound. Note that $\xi_f = \sum\limits_{k \in \mathscr{L}}\Phi_k e_f$. Hence $e_f$ is also bounded.
\end{proof}
\end{appendices}
% \vspace{-5mm}
\begin{table}[!h]
\centering
\caption{Description of attack signals.}
\small % Change font size
\renewcommand{\arraystretch}{0.2} % Adjust row height
\setlength{\tabcolsep}{1.5pt} % Reduce column separation
\begin{tabular*}{\columnwidth}{@{\extracolsep{\fill}}ccccc@{}}
\toprule
& \multicolumn{4}{c}{\textbf{\small $time (s)$}} \\ \cmidrule(l){2-5} 
\textbf{\small $attacks$} & \textbf{0-5} & \textbf{5-8} & \textbf{8-12} & \textbf{12-20} \\
\midrule
\textbf{$\mu_{f_1}$} & $0$ & $0.5$ & $(0.15t)^3+0.7$ & $e^{0.25t}+0.8$ \\
\textbf{$\mu_{f_2}$} & $0$ & $0.5$ & $(0.25t)^3+0.6$ & $e^{0.2t}+1$ \\
\textbf{$\mu_{f_3}$} & $0$ & $0.23$ & $(0.35t)^3+0.3$ & $e^{0.15t}+1.4$ \\
\textbf{$\mu_{f_4}$} & $0$ & $0.6$ & $(0.15t)^3+0.7$ & $e^{0.3t}+0.8$ \\
\midrule
\textbf{$\mu_{v_1}$} & $0$ & $2$ & $(0.35t)^3+2.1$ & $e^{0.3t}+3.2$ \\
\textbf{$\mu_{v_2}$} & $0$ & $1$ & $(0.45t)^3+1$ & $e^{0.25t}+3.5$ \\
\textbf{$\mu_{v_3}$} & $0$ & $2$ & $(0.25t)^3+2.1$ & $e^{0.35t}+2.6$ \\
\textbf{$\mu_{v_4}$} & $0$ & $1.5$ & $(0.15t)^3+1.5$ & $e^{0.45t}+1.7$ \\
\bottomrule
\end{tabular*}
\label{table1}
\end{table}

\ifCLASSOPTIONcaptionsoff
  \newpage
\fi
\vspace{-3mm}

\bibliography{ref123}

\end{document}